\documentclass[twoside,leqno]{article}
\usepackage[letterpaper]{geometry}
\usepackage{siamproceedings}

\usepackage[T1]{fontenc}
\usepackage{amsfonts}
\usepackage{graphicx}
\usepackage{epstopdf}
\usepackage{enumitem}
\usepackage{algorithmicx}
\usepackage[indLines=true]{algpseudocodex}
\ifpdf
  \DeclareGraphicsExtensions{.eps,.pdf,.png,.jpg}
\else
  \DeclareGraphicsExtensions{.eps}
\fi

\newsiamremark{remark}{Remark}
\newsiamthm{problem}{Problem}

\usepackage{amsopn}

\usepackage{array}
\newcolumntype{?}{!{\vrule width 2pt}}
\usepackage{csquotes}
\usepackage{amsmath}
\usepackage{amsfonts}
\usepackage{amssymb}
\usepackage{faktor}
\usepackage{bm}
\usepackage{booktabs}
\usepackage[normalem]{ulem}
\usepackage{graphics}
\usepackage[clock]{ifsym}
\usepackage{subcaption}
\usepackage{thm-restate}
\usepackage{hyperref}
\usepackage{cleveref}
\usepackage{xcolor}

\useunder{\uline}{\ul}{}

\newcommand{\Q}{\mathbb{Q}}
\newcommand{\R}{\mathbb{R}}

\newcommand{\Z}{\mathbb{Z}}
\newcommand{\N}{\mathbb{N}}

\renewcommand{\O}{\mathcal{O}}
\newcommand{\C}{\mathbb C}
\newcommand{\LRS}[1]{{\boldsymbol{#1}}}

\newcommand{\poly}{\mathsf{poly}}
\title{On the Complexity of the Skolem Problem at Low Orders}
\date{}
\author{Piotr Bacik\thanks{Department of Computer Science, Oxford University, UK and Max Planck Institute for Software Systems, Saarland Informatics Campus, Germany.} \and Jo\"el Ouaknine\thanks{Max Planck Institute for Software Systems, Saarland Informatics Campus, Germany.} \and James Worrell\thanks{Department of Computer Science, Oxford University, UK.}}

\begin{document}

\maketitle

\begin{abstract}
The Skolem Problem asks to determine whether a given linear recurrence sequence (LRS) $\langle u_n \rangle_{n=0}^\infty$ over the integers 
has a zero term, that is, whether there exists $n$ such that  $u_n = 0$.  Decidability of the problem is open in general,  with the most 
notable positive result being a decision procedure for LRS of order at most 4.

In this paper we consider a bounded version of the Skolem Problem, in which the input consists of an LRS $\langle u_n \rangle_{n=0}^\infty$ 
and a bound $N \in \mathbb N$ (with all integers written in binary), and the task is to determine whether there exists $n\in\{0,\ldots,N\}$ 
such that $u_n=0$.  We give a randomised algorithm for this problem that, for all $d\in \mathbb N$, runs  in polynomial time on the class
of  LRS of order at most $d$. As a corollary we show that the (unrestricted) Skolem Problem for LRS of order at most 4 lies in \textsf{coRP}, improving 
the best previous upper bound of $\mathsf{NP}^{\mathsf{RP}}$.  

The running time of our algorithm is exponential in the order of the LRS---a dependence that appears necessary in view of the 
$\mathsf{NP}$-hardness of the Bounded Skolem Problem. However, even for LRS of a fixed order, the problem involves detecting zeros within an 
exponentially large range. For this, our algorithm relies on results from $p$-adic analysis to
isolate polynomially many candidate zeros and then test in randomised polynomial time whether each candidate is an actual zero
by reduction to arithmetic-circuit identity testing.
\end{abstract}

\section{Introduction}

A \emph{linear recurrence sequence} (LRS) $\LRS{u} = \langle u_n \rangle_{n=0}^\infty$ is a sequence of integers
satisfying a {linear recurrence relation}:
\begin{align}\label{eqn:LRS_def}
u_{n+d} = a_{d-1} u_{n+d-1} + \dots + a_0 u_n \qquad (n\in\mathbb N)\, ,
\end{align}
where $a_0, \dots , a_{d-1} \in \Z$, and $a_0 \neq 0$.  The \emph{order} of $\LRS{u}$ is the smallest $d$ such that 
$\LRS{u}$ satisfies a relation of the form \eqref{eqn:LRS_def}.  Even though
LRS are ubiquitous in computer science, mathematics, and beyond, the decidability of 
many basic computational problems about LRS remain open. The most famous of these is the Skolem Problem~\cite{SalomaaS78,Tao08}:
\begin{problem}[The Skolem Problem]
Given an LRS $\LRS{u}$, does there exist $n \in \N$ such that $u_n = 0$?
\end{problem}
There is substantial interest in the decidability and complexity of the Skolem Problem, being closely connected with many topics in theoretical computer science and other areas, including loop termination \cite{Ouaknine_loop2015,almagor_deciding_2021}, formal power series (e.g. \cite{berstel_noncommutative_2010}, Section 6.4), matrix semigroups \cite{bell_mortality_2021}, stochastic systems \cite{agrawal_approximate_2015,barthe_universal_2020}, and control theory \cite{blondel_survey_2000}. 

The most significant structural result concerning the zeros of an LRS is the celebrated Skolem-Mahler-Lech Theorem \cite{Skolem_SML,Mahler_SML,lech_note_1953}, which states that the set of zeros of an LRS is a union of a finite number of arithmetic progressions and a finite set. The statement may be refined via the concept of \emph{non-degeneracy}. Define the \emph{characteristic polynomial} of the recurrence~\eqref{eqn:LRS_def} to be
\begin{align}
    g(X) := X^d - a_{d-1}X^{d-1} - \dots - a_0 \, .
\end{align}
Let $\lambda_1, \dots, \lambda_s \in \mathbb{\overline{\mathbb Q}}$ be the distinct roots of $g$; these are called  the \emph{characteristic roots} of $\LRS{u}$. We say $\LRS{u}$ is \emph{non-degenerate} if no ratio $\lambda_i/\lambda_j$ of distinct characteristic roots is a root of unity. A given LRS can be effectively decomposed as the interleaving of finitely many identically zero or non-degenerate LRS~\cite[Theorem 1.6]{recurrence_sequences}. The core of the Skolem-Mahler-Lech Theorem is that a non-degenerate LRS that is not identically zero has finitely many zero terms. The Skolem Problem therefore reduces to determining whether a non-degenerate LRS has a zero term.

One can also formulate a function problem variant of the Skolem Problem; find all $n \in \N$ such that $u_n = 0$. It is folklore that the function problem reduces to the decision problem of whether there exists a zero. Indeed, given a non-degenerate LRS $\LRS{u}$, one checks whether a zero exists using an oracle for the decision problem. If a zero $n_0$ exists, then one can find it by enumerating integers until it is found, and then shift the sequence to consider $\langle u_{n+n_0+1} \rangle_{n=0}^\infty$. One then repeats the process until the oracle finds no more zeros exist. Unfortunately, this reduction yields no information about how the complexities of the problems are related.

To date, only partial decidability results are known, either by restricting the order of the LRS or by restricting the structure of the characteristic roots. The breakthrough papers~\cite{Mignotte_distance,Vereshchagin} achieve decidability for LRS up to order 4 by exhibiting an effective bound $N$ such that if an LRS $\LRS{u}$ of order at most 4 has a zero term then it has a zero term $u_n=0$ where $n<N$.  It is shown in \cite{Orbit_problem_2016} that these methods yield an $\mathsf{NP}^\mathsf{RP}$ complexity upper bound for the Skolem Problem on LRS of order at most 4: specifically \cite[Appendix C]{Orbit_problem_2016} shows that the threshold $N$ has magnitude exponential in the description length of the LRS and that the existence of $n<N$ such that $u_n=0$ can be determined in $\mathsf{NP}^\mathsf{RP}$.  A procedure to substantiate the latter bound involves guessing $n$ and building an arithmetic circuit representing $u_n$ and then checking zeroness of $u_n$ in randomised polynomial time. Indeed, this is an instance of $\mathsf{EqSLP}$, which is the complete class for the following problem: given a division-free straight line program (or equivalently an arithmetic circuit) that outputs an integer $M$, determine whether $M = 0$. It is known that $\mathsf{EqSLP} \subseteq \mathsf{coRP}$~\cite{Schonhage_1979}.

The above considerations lead us to define the following problem:
\begin{problem}[The Bounded Skolem Problem]
Given an LRS $\LRS{u}$ and an integer $N>0$, does there exist $n \in \{0,\ldots,N\}$ such that $u_n = 0$? 
\end{problem}
\noindent
One can of course also formulate the function problem version; find all $n \in \{0,\dots,N\}$ such that $u_n = 0$. We can then reformulate the contribution of \cite{Orbit_problem_2016} as showing that the
Bounded Skolem Problem has complexity in $\mathsf{NP}^\mathsf{RP}$
and that the Skolem Problem for LRS of order 4 reduces in polynomial time to the bounded version.

In this paper we show that for each fixed $d\in \mathbb N$ the Bounded Skolem Problem for LRS of order at most $d$ lies in $\mathsf{coRP}$.  As a corollary we show that the Skolem Problem for LRS of order at most~4 also lies in $\mathsf{coRP}$, improving the previous bound 
of $\mathsf{NP}^\mathsf{RP}$.
Since the threshold $N$ in the Bounded Skolem Problem is encoded in binary, the range $\{0,\ldots,N\}$ in which a zero is sought has cardinality exponential in the size of the input.  Nevertheless, we give a randomised algorithm that can detect the existence of a zero in polynomial time for LRS of a fixed order.  The running time of our algorithm is exponential in the order of the LRS\@. The exponential dependence on the order seems unavoidable, since the $\mathsf{NP}$-hardness proof of the Skolem Problem straightforwardly carries over to its bounded version, although it no longer applies if the order of the LRS is fixed~\cite{Akshay2017,Blondel2002}.\footnote{The \textsf{NP}-hardness 
proof for the Skolem Problem in~\cite[Theorem 6]{Akshay2017} is by  reduction from Subset Sum.
The reduction maps an instance of Subset Sum to an
LRS that has a zero iff
the instance of Subset Sum has a solution.   The constructed LRS has period
the product $N:=p_1\cdots p_m$ of the first $m$ primes, where $m$ is the number of integers in the instance of Subset Sum.
Evidently the same reduction works also for the Bounded Skolem Problem.}

Our algorithm for the Bounded Skolem Problem is based on a $p$-adic approach that was introduced by Skolem in his original proof of the Skolem-Mahler-Lech Theorem.  This involves extending an integer LRS $\boldsymbol u = \langle u_n\rangle_{n=0}^\infty$ to a $p$-adic analytic function $F:\mathcal O_p\rightarrow \mathcal O_p$, for a suitable prime $p$, such that $F(n)=u_n$ for all $n\in \mathbb N$.  Here, $\mathcal O_p$ is the valuation ring of the field $\mathbb C_p$ which is the $p$-adic analogue of the field of complex numbers.

As we explain below,
while it is not known how to decide whether an LRS has an integer zero, one can decide in polynomial time whether its extension $F$ has a zero in $\mathcal O_p$. 
More generally, by considering zeros of the power series $F(p^r x + m)$ for given $r\in \mathbb N$ and $m \in \{0,\ldots,p^r-1\}$, 
one can determine whether $F$ has a zero $a \in \mathcal O_p$ such that $a\equiv m\bmod{p^r}$.

Our procedure performs a depth-first search of all residues $m \in \{0,\ldots,p^r-1\}$, for successively larger $r$, such that $F$ has a zero in $\mathcal O_p$ whose residue modulo $p^r$ is equal to $m$.
Assume without loss of generality that the input $N$ to the Bounded Skolem Problem has the form  $N=p^n-1$.  We call each residue $m\in\{0,\ldots,N\}$ for which
there exists $a \in \mathcal O_p$ with $F(a)=0$ and $a\equiv m\bmod (N+1)$
 a \emph{candidate zero} of $\boldsymbol u$.  Every integer zero of $\boldsymbol u$ in the interval $\{0,\ldots,N\}$ is a candidate zero and, conversely, every candidate zero
 is the residue modulo $N+1$ of a zero of $F$ in $\mathcal O_p$.  For each candidate zero $m\in \{0,\ldots,N\}$, we determine whether $u_m$ is zero in randomised polynomial time by
 constructing an arithmetic circuit representing $u_m$, using iterated matrix powering, and checking the circuit for zeroness using standard methods in arithmetic circuit identity testing.

The overall polynomial-time bound of the procedure (for each fixed recurrence depth $d$) crucially relies on a quantitative refinement of Skolem's proof, due to Van der Poorten and Schlickewei~\cite{van_der_poorten_1991}.
We use this result in two different ways.  First it gives an upper bound on the number of zeros of $F$ in $\mathcal O_p$, which translates in our setting to a polynomial bound on the number of candidate zeros that our algorithm examines.  In other words, we need only search a small fragment of the exponentially large search tree of residues modulo $N+1$.  Secondly, we use the result of \cite{van_der_poorten_1991} in combination with the Weierstrass Preparation Theorem  (which characterises the number of roots of a power series in terms of the absolute values of its coefficients) to decide in polynomial time whether the power series $F$ has a zero in $\mathcal O_p$.  The key issue in this case is to obtain a polynomial bound on the number of coefficients of a power series that must be examined to determine whether it has a root in $\mathcal O_p$.

Although the correctness of our algorithm is based on $p$-adic techniques, the implementation works with the underlying LRS and does not directly compute with $p$-adic numbers.  This is because we work with the so-called Mahler-series representation of $p$-adic power series, which are based on the Mahler polynomials $\binom{x}{k} = \frac{x(x-1)\cdots(x-k+1)}{k!}$ for $k\in\mathbb N$.  Conveniently for our purposes, the interpolation $F$ of an LRS has a Mahler series expansion $F(x) = \sum_{k=0}^\infty \beta_k \binom{x}{k}$ whose coefficients $\beta_k$ lie in $\mathbb Z$ and can, moreover, be directly calculated from the LRS by a simple formula involving the discrete difference operator.  The critical information for applying the $p$-adic Weierstrass Preparation Theorem for locating the roots of $F$ can be extracted from its Mahler series expansion.

\section{Preliminaries} 
\label{sec:prelims}
\subsection{Linear recurrence sequences}
We recall some basic notions about linear recurrence sequences; more details may be found in \cite{recurrence_sequences}. 
We also establish some notation. Let $\LRS{u}$ be an LRS of order $d$ satisfying the relation \eqref{eqn:LRS_def}.
%and let $\lambda_1, \dots, \lambda_s$ be the characteristic roots of $\LRS{u}$. 
The \emph{companion matrix} of $\LRS{u}$ is defined to be
\begin{align*}
 A = \begin{pmatrix} a_{d-1} & \dots & a_1 & a_0 \\ 1 & \dots & 0 & 0 \\ \vdots & \ddots & \vdots & 0 \\ 0 & \dots  & 1 & 0  \end{pmatrix} \, .
\end{align*}
Writing $\alpha = \begin{pmatrix} 0 & 0 \dots & 1 \end{pmatrix}$, $\beta = \begin{pmatrix} u_{d-1} & u_{d-2} \dots & u_0 \end{pmatrix}^T$, we get $u_n = \alpha A^n \beta$. 
%The eigenvalues of $A$ are exactly the characteristic roots $\lambda_1, \dots, \lambda_s$.
We denote by $\|\LRS{u}\|$ the \emph{size} of the representation of $\LRS{u}$, defined as
\begin{align*}
\|\LRS{u}\| = \sum_{i=0}^{d-1} (\|a_i\| + \|u_i\|)
\end{align*}
where $\|y\|$ denotes the number of bits needed to represent the integer $y$. Throughout the paper, let $\poly(\|\LRS{u}\|)$ denote any quantity that is bounded above by $\|\LRS{u}\|^{O(1)}$.
%It is well known that $\LRS{u}$ admits an exponential-polynomial representation of the form 
%\begin{align} \label{eqn:binet}%
%u_n = \sum_{i=1}^s P_i(n) \lambda_i^n
%\end{align} 
%where each $P_i$ is a polynomial with degree one less than the multiplicity of $\lambda_i$ as a root of the characteristic polynomial of $\LRS{u}$. 
\subsection{\texorpdfstring{$p$}{p}-adic numbers}
We briefly recall relevant notions about $p$-adic numbers.  We refer to~\cite{robert_course_2000,gouvea_p-adic_2020} for more details.
Given a prime number $p$, every non-zero rational number $x \in \Q$ may be written as $x = \frac{a}{b}p^r$ for some integers $a,b$ coprime to $p$, with $b$ non-zero, and $r \in \Z$. We define the valuation $v_p(x) = r$, and $v_p(0) = \infty$.  From this derivation we see that $v_p$ satisfies the ultrametric inequality
$v_p(x+y) \geq \min(v_p(x),v_p(y))$ for all $x,y \in \mathbb Q$.

We define an absolute value on $\Q$ by $|x|_p = p^{-v_p(x)}$.  The ultrametric inequality on $v_p$ translates to the 
strong triangle inequality: $|x+y|_p \leq \max(|x|_p,|y|_p)$ for all $x,y \in \mathbb Q$.
In the same way as one  completes $\Q$ with respect to the standard absolute value $|\cdot|$ to get $\R$, one may complete $\Q$ with respect to $|\cdot|_p$ to get the $p$-adic numbers $\Q_p$. If one completes $\Z$ with respect to $|\cdot|_p$ then one gets the $p$-adic integers $\Z_p$, which may also be defined as the unit disc in $\Q_p$:
\begin{align*}
\Z_p = \{x \in \Q_p : |x|_p \leq 1\} = \{x \in \Q_p : v_p(x) \geq 0 \} \, .
\end{align*}

The absolute value $|\cdot|_p$ on $\mathbb Q_p$ extends uniquely to an absolute value on the algebraic closure $\overline{\mathbb Q_p}$.  Specifically, given $\alpha \in \overline{\mathbb Q_p}$, we define $|\alpha|_p := |N(\alpha)|_p^{1/n}$, where $n$ is the degree of $\alpha$ over $\mathbb Q_p$ and $N(\alpha)$ is the \emph{norm} of $\alpha$, that is the determinant of the $\mathbb Q_p$-linear transformation 
$\mu_\alpha : \mathbb Q_p(\alpha)\rightarrow \mathbb Q_p(\alpha)$ given by $\mu_\alpha(x)=\alpha x$.
We correspondingly extend the valuation $v_p(\cdot)$ to $\overline{\mathbb Q_p}$ by $v_p(\alpha) := \frac{1}{n}v_p(N(\alpha))$.  With these definitions the identity 
$|\alpha|_p=p^{-v_p(\alpha)}$ holds for all $\alpha \in \overline{\mathbb Q_p}$.  While $\overline{\mathbb Q_p}$ is not complete with respect to $|\cdot|_p$, taking the Cauchy completion we obtain a field $\mathbb C_p$ that is both algebraically closed and complete under the extension of $|\cdot|_p$ to $\mathbb C_p$.

Define the closed disc $\overline D(z,r)$ centred at $z \in \mathbb C_p$ with valuation $r$ to be
\begin{align} \label{eqn:disc_def}
\overline D(z,r) = \{x \in \C_p : v_p(x-z) \geq r \} = \{x \in \C_p : |x-z|_p \leq p^{-r}\}
\end{align}
and denote the unit disc of $\C_p$ by \[\O_p = \overline D(0,0)=\left\{x\in\C_p:v_p(x)\geq 0\right\} = \{x \in \C_p : |x|_p \leq 1\}\, .\]  Note that $\mathcal O_p$ is a subring of $\C_p$.
\subsection{\texorpdfstring{$p$}{p}-adic power series}
Let $\langle b_j\rangle_{j=0}^\infty$ be a sequence in $\mathbb Z_p$.  A sufficient and necessary condition for the power series $F(x)=\sum_{j=0}^\infty b_j x^j$
to converge on $\mathbb Z_p$ is that $v_p(b_j) \rightarrow \infty$ as $j\rightarrow \infty$.  A \emph{$p$-adic analytic function} $F:\mathbb Z_p\rightarrow \mathbb Z_p$
is one that is given by a convergent power series.  For a $p$-adic analytic function $F$ and $\alpha \in \mathcal O_p$ it is the case that $F(\alpha)$ converges to an element 
of $\mathcal O_p$.  Thus we may refer to the zeros of $F$ in $\mathcal O_p$.

The following theorem, found in \cite[Theorem 7.2.6]{gouvea_p-adic_2020}, is fundamental to our approach.

\begin{theorem}[$p$-adic Weierstrass Preparation Theorem] \label{thm:Weierstrass}
Let $F(x) = \sum_{j=0}^\infty b_j x^j \in \Z_p [\![x]\!]$ be a non-zero convergent power series. Suppose $j_0$ is an integer such that
\begin{align*}
v_p(b_{j_0}) = \underset{j} \min\, {v_p(b_j)} \hspace{1cm} \text{and} \hspace{1cm} v_p(b_{j_0}) < v_p(b_j) \hspace{0.5cm}  \forall j > j_0 \, .
\end{align*}
Then there is a polynomial $g(x)$ of degree $j_0$ and a power series $h(x)$ with no zeros in $\O_p$ such that
\begin{align*}
    F(x) = g(x) h(x) \, .
\end{align*}
In particular, $F$ has exactly $j_0$ zeros in $\O_p$, counting multiplicity.
\end{theorem}

To facilitate the application of this theorem, we formulate a further definition and corollary.
\begin{definition}
Let $F$ and $j_0$ be as in Theorem \ref{thm:Weierstrass}. Define $V(F) := \underset{j} \min \, v_p(b_j) = v_p(b_{j_0})$, and write $j(F) := j_0$.
\end{definition}
Equivalently, $j(F)$ is the largest index of a power series coefficient of $F$ with minimal valuation among all coefficients, and $V(F)$ is this minimal valuation.
\begin{corollary} \label{cor:Weierstrass}
Let $F$ be as in Theorem \ref{thm:Weierstrass}. Define $h_{z,r}(x) = F(p^rx + z)$. Then $F$ has exactly $j(h_{z,r})$ zeros in the disc $\overline D(z,r)$.
\end{corollary}
\begin{proof}
By Theorem \ref{thm:Weierstrass}, $h_{z,r}$ has exactly $j(h_{z,r})$ zeros in $\O_p$. Note that $x \in \O_p = \overline D(0,0)$  if and only if $p^r x + z \in \overline D(z,r)$, so $F$ has exactly $j(h_{z,r})$ zeros in the disc $\overline  D(z,r)$.
\end{proof}

\subsection{Mahler Series}
The computation of $j(F)$ for $p$-adic analytic functions $F$ arising from LRS plays a critical role in our zero-finding procedure.
The following two sections show how this can be done directly from the LRS, without the need to compute the coefficients of $F$.
For this we need some background on Mahler series. The following may be found in \cite[Section IV 2.3-2.4, Section VI 4.7]{robert_course_2000}.

%\begin{definition}
Define the difference operator $\Delta$ by $(\Delta F)(x) = F(x+1)-F(x)$ 
for any function $F : \Z_p \rightarrow \Z_p$.
%\end{definition}
For $k\geq 0$, consider the polynomial
\[\binom{x}{k}:=\frac{x(x-1)\cdots(x-k+1)}{k!}\, . \] 
Since $\binom{n}{k} \in \mathbb Z$ for all $n \in \mathbb Z$ and  $\mathbb Z$ is dense in $\Z_p$ with respect to the 
topology induced by $|\cdot|_p$,
it follows that $\binom{x}{k}$ maps $\mathbb Z_p$ to $\mathbb Z_p$.
A \emph{Mahler series} is a formal series 
\begin{gather}
    F(x) = \sum_{k=0}^\infty \beta_k \binom{x}{k} \, ,
\label{eq:mahler}
\end{gather} 
where $\beta_k \in \Z_p$ is such that $v_p(\beta_k)\rightarrow \infty$
as $k\rightarrow \infty$.  Such a series defines a 
continuous function $F : \mathbb Z_p \rightarrow \mathbb Z_p$ and, conversely,
every continuous  function $F : \mathbb Z_p \rightarrow \mathbb Z_p$ can be represented by a Mahler series 
of the form~\eqref{eq:mahler} in which 
\begin{gather}
    \beta_k = (\Delta^k F)(0) = \sum_{j=0}^k (-1)^{k-j} \binom{k}{j} F(j) \, .
\label{eq:deriv-formula}
\end{gather} 
Under the further assumption that $v_p(\beta_k/k!) \to \infty$ as $k \to \infty$, the Mahler series~\eqref{eq:mahler}
defines an analytic function $F : \mathbb Z_p \rightarrow \mathbb Z_p$, that is,
$F$ admits a power series expansion
$F(x) = \sum_{j=0}^\infty b_j x^j$ where $b_j \in \mathbb Q_p$ for all $j$ and $v_p(b_j) \rightarrow \infty$. 

The following proposition characterises $j(F)$ for an analytic function $F$ in terms of its Mahler series.

\begin{proposition} \label{prop:coef_to_mahler}
Suppose $F : \Z_p \to \Z_p$ is an analytic function whose respective power series and Mahler series representations are:
\begin{align} \label{eqn:Mahler_and_power}
F(x) = \sum_{j=0}^\infty b_j x^j = \sum_{k=0}^\infty \beta_k \binom{x}{k} \,
\end{align}
Define $V'(F) := \underset{k \geq 0} \min\, v_p(\beta_k/k!)$ and define $k(F)$ to be the  largest integer $k$ such that $v_p(\beta_k/k!) = V'(F)$. Then $V(F) = V'(F)$ and $ j(F) = k(F) $.
\end{proposition}
\begin{proof}
We have 
\begin{align} \label{eqn:basis_change}
k!\binom{x}{k} = \sum_{j=0}^k (-1)^{k-j} \begin{bmatrix} k \\ j\end{bmatrix}x^j \hspace{5mm} \text{and} \hspace{5mm} x^j = \sum_{k=0}^j \begin{Bmatrix} j \\ k \end{Bmatrix} k!\binom{x}{k}
\end{align}
where $\begin{bmatrix} k \\ j \end{bmatrix}$ and $\begin{Bmatrix} j \\ k \end{Bmatrix}$ are positive integers known as the Stirling numbers of the first and second kind respectively. 
By combining \eqref{eqn:basis_change} with \eqref{eqn:Mahler_and_power} we see that
\begin{align*}
b_j = \sum_{k=j}^\infty (-1)^{k-j} \frac{\beta_k}{k!} \begin{bmatrix} k \\ j \end{bmatrix} \hspace{5mm} \text{and} \hspace{5mm} \frac{\beta_k}{k!} = \sum_{j=k}^\infty b_j \begin{Bmatrix} j \\ k \end{Bmatrix}
\end{align*}
which gives
\begin{align*}
    V(F) = v_p(b_{j(F)}) = v_p \left(\frac{\beta_{j(F)}}{j(F)!}\right) \geq V'(F)
\end{align*}
and
\begin{align*}
 V'(F) = v_p \left( \frac{\beta_{ k(F)}}{ k(F)!} \right) = v_p(b_{k(F)}) \geq V(F) \, .
\end{align*}
The combination of these inequalities give $V(F) = V'(F)$ and $j(F) = k(F)$. 

\end{proof}

\subsection{%\texorpdfstring{$p$}
Interpolation of Linear Recurrence Sequences}
\label{sec:interpolate}
It is well known that for a suitable prime $p$,
a linear recurrence sequence admits a decomposition into finitely many subsequences such that each subsequence 
can be interpolated to a $p$-adic analytic function.  This fact is the basis of the Skolem's original proof of the Skolem-Mahler-Lech theorem.  
In this section, we give a self-contained account of this construction.  Furthermore,
for each of the $p$-adic analytic functions $F$ that arise from the interpolation of an LRS we give a characterisation of $j(F)$ in terms of the
first $d-1$ discrete derivatives of the LRS at zero, where $d$ is the order of the recurrence.
In the process, we recover a result of \cite[Lemma 2]{van_der_poorten_1991} that $j(F)\leq d-1$ 
for $p$ sufficiently large relative to the order~$d$ of the recurrence:

\begin{theorem}
  Let $\LRS{u}$ satisfy the order-$d$ linear recurrence~\eqref{eqn:LRS_def}
  and let $p > d+1$ be a prime not dividing the constant term $a_0$ of the recurrence.  Then there exists $M < p^{d^2}$ and analytic functions $F_\ell : \Z_p\rightarrow \Z_p$,
$\ell \in \{0,\ldots,M-1\}$, such that for all $\ell$ we have 
\begin{enumerate}
    \item $F_{\ell}(n) = u_{Mn+\ell}$ for all $n\in \mathbb N$;
    \item $V(F_\ell) = \min \{ v_p((\Delta^j F_\ell)(0)) : 0 \leq j \leq d-1\}$;
    \item $j(F_\ell) =\max \{ j \in \{0,\ldots,d-1\} : v_p((\Delta^j F_\ell)(0)) = V(F_\ell)\}$.
\end{enumerate}
\label{thm:interpolate}
\end{theorem}
\begin{proof}
Write $u_n = x^\top A^n y$, where $A$ is the companion matrix of the recurrence and 
\[ x = \begin{pmatrix} 0 & 0  & \dots & 1 \end{pmatrix},\qquad \qquad  y = \begin{pmatrix} u_{d-1} & u_{d-2}&  \dots & u_0 \end{pmatrix}^T \, .\]
Since $p \nmid a_0$, we have $p \nmid \det(A)$ and hence $A$ is invertible in $\mathrm{GL}_d(\mathbb Z/p\mathbb Z)$.
It follows that 
$A^M \equiv I \pmod p$, for some $0<M<p^{d^2}$.
Write $A^M = I+pB$ for some $d\times d$ integer matrix $B$.  Then for $\ell \in \{0,\ldots,M\}$ we have
\begin{eqnarray*}
u_{Mn+\ell} &=& x^T A^{Mn+\ell} y \\
    &=& x^T (I+pB)^{n} A^\ell y \\
    &=& \sum_{k=0}^n \binom{n}{k} p^k \, x^T B^k A^\ell y \\
    &=& \sum_{k=0}^\infty \beta_{\ell,k} \binom{n}{k} \, ,
\end{eqnarray*}
where $\beta_{\ell,k} := p^k \, x^T B^k A^\ell y$.  Since $v_p(\beta_{\ell,k}) \geq k$, using the standard bound $v_p(k!) \leq \frac{k}{p-1}$,
we have $v_p(\beta_{\ell,k}/k!) \geq \frac{k(p-2)}{p-1}$.  Thus $v_p(\beta_{\ell,k}/k!)\rightarrow \infty$ as $k\rightarrow \infty$
and so the Mahler series $F_\ell (x):= \sum_{k=0}^\infty  \beta_{\ell,k} \binom{x}{k}$ defines an analytic function 
such that $F_\ell(n) = u_{Mn +\ell}$ for all $n \in \mathbb Z$.  

Fix $\ell \in \{0,\ldots,M-1\}$ and 
write $\gamma_{\ell,k} := x^\top B^kA^\ell y$, so that $\beta_{\ell,k} = p^k \gamma_{\ell,k}$. 
By the Cayley-Hamilton theorem, there are integers $b_0,\ldots,b_{d_1}$ such that $B^d = b_0 B^{d-1} +\cdots b_{d-1}I$. Hence
\begin{gather} 
\gamma_{\ell,k} = x^\top B^k A^\ell y = \sum_{i=0}^{d-1} b_i \, x^\top B^{k-1-i} A^\ell y = \sum_{i=0}^{d-1} b_i \gamma_{\ell,k-1-i}  
\label{eq:RECUR}
\end{gather}
for all $k \geq d$.
Let $m:=\min \left\{ v_p(\gamma_{\ell,k}) : 0 \leq k \leq d-1 \right\}$.  Then, by the Equation~\ref{eq:RECUR},
it follows by strong induction on $k$ that $v_p(\gamma_{\ell,k}) \geq m$ for all $k\in\mathbb N$.  

For all $k \geq d$, using the standard bound $v_p(k!) \leq \frac{k}{p-1}$,  we  have:
\begin{eqnarray*}
    v_p(\beta_{\ell,k}/k!) &=&  v_p(\gamma_{\ell,k}) + v_p(p^k/k!)\\
       &\geq &  m+  \textstyle\frac{k(p-2)}{p-1}  \\
       &\geq &   m+  \textstyle\frac{d(p-2)}{p-1}  \\
       &>& m+d-1 \quad\text{(since $p>d+1$).}
\end{eqnarray*}
On the other hand, for $0\leq k \leq d-1$ we have $v_p(\beta_{\ell,k}/k!) \leq v_p(\gamma_{\ell,k})+d-1$ and hence
\[
\min_{0\leq k\leq d-1} v_p(\beta_{\ell,k}/k!) \leq m+d-1 \, . 
\]
We conclude that  $k(F_\ell) \leq d-1$. 

Applying Proposition~\ref{prop:coef_to_mahler}, since $d-1<p$, we have 
\[ V(F_\ell)=V'(F_\ell)=\min_{0\leq k\leq d-1} v_p(\beta_{\ell,k} /k!) = \min_{0\leq k\leq d-1} v_p(\beta_{\ell,k}) = \min_{0\leq k\leq d-1}v_p((\Delta^k F_\ell)(0))\]
and 
$j(F) = k(F)$ is the maximum value of $k \in \{0,\ldots,d-1\}$ such that $v_p(\beta_{\ell,k}) = V'(F)$.  
\end{proof}

\section{Algorithm and Complexity Analysis}
Let $\mathrm{LRS}(d)$ denote the class of LRS of order $d$. Throughout this section, we assume we are working inside the class $\mathrm{LRS}(d)$ for some $d \in \mathbb N$.
Our complexity estimates are in terms of $\| \boldsymbol u \|$, the length of the description of initial values and recurrence 
that define $\boldsymbol{u}$, and the bit length of the threshold $N$ describing the range $\{0,\ldots,N\}$ in which we seek zeros of LRS\@.

The goal of this section is to prove the following result:
\begin{restatable}{theorem}{SkolemEqSLP} \label{thm:Skolem_EqSLP} 
For every $d \in \mathbb N$, the Bounded Skolem Problem on $\mathrm{LRS}(d)$ is disjunctively 
Turing reducible to $\mathsf{EqSLP}$ and hence lies in $\mathsf{coRP}$.
\end{restatable}
Note above that 
one can test whether  at least one among a finite collection of arithmetic circuits is zero 
by testing the product for zeroness, which can be done in \textsf{coRP}~\cite{Schonhage_1979}.
\subsection{Informal outline of the algorithm}
We will first outline the general ideas behind the algorithm and the ingredients required to prove the required complexity bounds on the Bounded Skolem Problem. Given an LRS $\LRS{u}$ and integer upper bound $N$ written in binary, let $\|\LRS{u}\|$ and $\|N\|$ denote the lengths of their respective binary strings. Recall we are required to search for integers $0 \leq n \leq N$ with $u_n = 0$.

In short, the idea is to find $p$-adic approximations of any possible zeros of $\LRS{u}$ by finding discs of decreasing size in which zeros may lie. We do this approximation until for each disc there is only one possible integer that lies below the upper bound $N$ and lies inside the disc, and then we check all the candidate integer zeros. In more detail, recalling the definition \eqref{eqn:disc_def} of the disc $\overline D(z,r)$, we do the following:

\begin{enumerate}
    \item Choose a prime $p$ that does not divide the last coefficient $a_0$ of the recurrence.
    \item Split the LRS into subsequences $u_{Mn+\ell}$ for some integer $M$ and $0 \leq \ell \leq M-1$ such that $u_{Mn+\ell}$ may be interpolated by a $p$-adic analytic function $F_\ell : \Z_p \to \Z_p$. We have $F_\ell$ is identically zero if and only if $u_{Mn+\ell} = 0$ for $n=0,1,\dots,d-1$. This may be checked in polynomial-time as $Mn+\ell = \poly(\|\LRS{u}\|)$ for these values of $n$. We analyse those functions $F_\ell$ for which the corresponding sequence $u_{Mn+\ell}$ is not identically zero.
    
    \item For $R$ such that $Mp^R + \ell \geq N$, find all discs $\overline D \left(z,R \right)$ for $0 \leq z \leq p^R -1$ containing a zero $x_0 \in \O_p$ of $F_\ell$. This identifies all residue classes mod $p^R$ that could contain an integer zero of $u_{Mn+\ell}$.
    \begin{itemize}
        \item To do this efficiently, we inductively find each disc $\overline D(z,r)$ containing a zero of $F_\ell$ for each $0 \leq r \leq R$ and $0 \leq z \leq p^r-1$.
        \item If we have found that the disc $\overline D(z,r)$ contains a zero of $F_\ell$ for integer $0 \leq z \leq p^r-1$, then check each disc $\overline D(z+p^ra,r+1)$ to determine whether it has any zeros of $F_\ell$, for integers $0 \leq a \leq p-1$. 
        \item Checking whether the disc $\overline D(z,r)$ contains any zeros of $F_\ell$ may be done by computing $j(h_{z,r})$, where $h_{z,r}(x) = F(p^r x + z)$, by Corollary \ref{cor:Weierstrass}. In practice we do this by looking at the $p$-adic valuations of the Mahler-series coefficients $\beta_{\ell,k}$
        for each $k$ (as defined in Section~\ref{sec:interpolate}), using Theorem \ref{thm:interpolate}.
    \end{itemize}
    \item Check whether $u_{Mz+\ell} = 0$ for each disc $\overline D(z,R)$ computed.
\end{enumerate}
Correctness of the algorithm is straightforward to see as by definition of $R$, any disc $\overline{D}(z,R)$ contains at most one integer $n$ such that $0 \leq Mn+\ell \leq N$, namely $n=z$ being the only possible candidate. The discs found by the end of the computation are the only discs $\overline D(z,R)$ containing a $p$-adic zero $x_0 \in \O_p$ of $F_\ell$, so they obviously contain all integer zeros of $F_\ell(n) = u_{Mn+\ell}$. The bulk of the problem is to prove that steps 1-3 above may be done in $\poly(\|\LRS{u}\|,\|N\|)$ time. The strategy of the next subsection will be to show that:
\begin{enumerate}
    \item $p,M$ may be taken to be $\poly(\|\LRS{u}\|)$ in magnitude.
    \item Computing whether the disc $\overline D(z,r)$ contains any zeros of $F_\ell$ takes $\poly(\|\LRS{u}\|,\|N\|)$ time for any integer $0 \leq z \leq p^r-1$ and $1 \leq r \leq R$ for $R$ smallest integer such that $Mp^R + \ell \geq N$
    \item It takes $\poly(\|\LRS{u}\|,\|N\|)$ many of the above computations to determine all the residue classes mod $p^R$ possibly containing a zero of $F_\ell$
\end{enumerate}
\subsection{Proofs of complexity bounds}
In this subsection we prove that all the parts of the algorithm outlined have the required complexity bounds.
\begin{lemma} \label{lem:prime_and_M}
We may find a prime $p > d+1$ and integer $M$ with $\poly(\|\LRS{u}\|)$ magnitude in $\poly(\|\LRS{u}\|)$ time that satisfy the conditions of Theorem \ref{thm:interpolate}.
\end{lemma}
\begin{proof}
Choose any prime $p \geq d+2$ such that $p \nmid a_0$, the last coefficient of the recurrence relation \eqref{eqn:LRS_def} that $\LRS{u}$ satisfies. We justify that such a prime may be found, of magnitude $\poly(\|u\|)$. Define Chebyshev's prime counting function 
\begin{align*}
    \vartheta(n) = \sum_{\substack{p \leq n \\ p \text{ prime}}} \log p \, .
\end{align*}
Then $e^{\vartheta(n)}$ is the product of all primes at most $n$. If $e^{\vartheta(n)} > (d+1)!|a_0|$, then there must be a prime $p$ of magnitude at most $n$ such that $p \nmid (d+1)!|a_0|$, which satisfies our requirements. By \cite[Corollary to Theorem 4]{rosser_1962}, we have $\vartheta(n) > n \left(1- \frac{1}{\log n} \right)$ for all $n \geq 41$, so in particular there must be a suitable prime $p$ with $p \leq \max \left\{41,\frac{\log((d+1)!|a_0|)}{1 - 1/\log 41} \right\} = \poly(\|u\|)$. 

Now, by Theorem \ref{thm:interpolate} we can take $M < p^{d^2}$, in particular the smallest positive integer such that $A^M \equiv I \mod p$, where $A$ is the companion matrix of $\LRS{u}$. Such $M$ is easily found in $\poly(\|\LRS{u}\|)$ time by directly checking if $A^m \equiv I \mod p$ for each $0 < m < p^{d^2}$.
\end{proof}
We would now like to prove that the search for $p$-adic discs containing zeros of $F_\ell$ takes $\poly(\|\LRS{u}\|)$ time. To do this, we need to be able to find $j(h_{z,r})$ quickly, where $h_{z,r}(x) = F(p^r x + z)$, so we need to show $V(h_{z,r})$ is not too large.
\begin{lemma} \label{lem:min_coef}
Let $F_\ell : \mathbb Z_p \to \mathbb Z_p$ interpolate the subsequence $u_{Mn+\ell}$ for $p,M$ as in Lemma \ref{lem:prime_and_M}, and suppose that the subsequence $u_{Mn+\ell}$ is not identically zero. Let $h_{z,r}(x) = F_\ell(p^rx+z)$ for integers $r$ and $0 \leq z \leq p^r-1$. Then 
\begin{enumerate}
    \item $j(h_{z,r}) \leq j(F_\ell) \leq d-1$
    \item $V(h_{z,r}) \leq V(F_\ell) + rj(F_\ell) = \poly(\|\LRS{u}\|,r)$.
\end{enumerate}
\end{lemma}
\begin{proof}
We have $j(h_{z,r}) \leq j(F_\ell)$ by Corollary \ref{cor:Weierstrass} and we have $j(F_\ell) \leq d-1$ by Theorem \ref{thm:interpolate}. 
This proves Item 1.

For the bound on $V(h_{z,r})$ in Item 2., note that by Item 2 of Theorem~\ref{thm:interpolate} we have

\begin{align*}
V(F_\ell) \in \{ v_p((\Delta^0F_\ell)(0)), v_p((\Delta^1F_\ell)(0)), \dots, v_p((\Delta^{d-1}F_\ell)(0)) \} \, .
\end{align*}
Then from the formula 
\begin{align*}
    (\Delta^k F_\ell)(0) = \sum_{j=0}^k (-1)^{k-j} \binom{k}{j} F_\ell(j)
\end{align*}
we have
\begin{align*}
V(F_\ell) \leq \log \left( \underset{0 \leq k \leq d-1} \max |(\Delta^k F_\ell)(0)| \right) \leq \log \left( \sum_{j=0}^{d-1} \binom{d-1}{j} |F_\ell(j)| \right) \leq \log\left( \sum_{j=0}^{d-1} 2^d |u_{Mj + \ell}| \right)
\end{align*}
and using that $|u_{Mj+\ell}| \leq 2^{(Mj+\ell+1)\|\LRS{u}\|}d^{Mj+\ell}$, which is easily shown by induction with the recurrence relation \eqref{eqn:LRS_def}, we get $V(F_\ell) = \poly(\|\LRS{u}\|)$.

For the bound on $V(h_{z,r})$, note the following. For any analytic $F(x) = \sum_{j=0}^\infty b_j x^j$, for $0 \leq k \leq d-1$ we have $F^{(k)}(0) = b_kk!$ so since $p \geq d+2$ we have $v_p(F^{(k)}(0)) = v_p(b_k)$. Furthermore, we have $v_p(b_{j(F)}) < v_p(b_k)$ for all $k > j(F)$ so $v_p(F^{(j(F))}(z)) = v_p(F^{(j(F))}(0))$ for all $z \in \Z$. From this we can deduce
\begin{align*}
V(h_{z,r}) = v_p(h_{z,r}^{(j(h_{z,r}))}(0)) \leq v_p(h_{z,r}^{(j(F_\ell))}(0)) = v_p(p^{rj(F_\ell)}F_\ell^{(j(F_\ell))}(z)) &= V(F_\ell) + rj(F_\ell)
\end{align*}
and we have $V(F_\ell) + rj(F_\ell) \leq V(F_\ell) + r(d-1) = \poly(\|\LRS{u}\|,r)$.

\end{proof}
Having shown that $V(h_{z,r})$ cannot be too large, we now show we can compute $j(h_{z,r})$ sufficiently quickly.
\begin{proposition} \label{prop:fast_disc}
Let $F_\ell : \mathbb Z_p \to \mathbb Z_p$ interpolate a not-identically-zero subsequence $u_{Mn+\ell}$ with $p,M = \poly(\|\LRS{u}\|)$. For any integers $r>0$ and $0 \leq z \leq p^r-1$, we may find the number of zeros (counted with multiplicity) of $F_\ell$ in the disc $\overline D(z,r)$ in $\poly(\|\LRS{u}\|,r)$ time.
\end{proposition}
\begin{proof}
By Corollary \ref{cor:Weierstrass} it suffices to find $j(h_{z,r})$, which by Lemma \ref{lem:min_coef} satisfies $j(h_{z,r}) \leq d-1$. By Proposition \ref{prop:coef_to_mahler} and Lemma \ref{lem:min_coef}, $j(h_{z,r})$ is simply the largest integer $k \in \{0, \dots, d-1\}$ such that $v_p((\Delta^k h_{z,r})(0))$ is minimal (and equal to $V(h_{z,r})$). By Lemma \ref{lem:min_coef} we have that $V(h_{z,r}) < \nu$ for some $\nu = \poly(\|\LRS{u}\|,r)$.\footnote{By inspecting the proof of Lemma \ref{lem:min_coef} we may take $\nu = \log(2^d2^{(Md+\ell+1)\|\LRS{u}\|}d^{Md+\ell + 1}) + rd+1$ for actual computations.} 

Then at least one of $(\Delta^0h_{z,r})(0), (\Delta^1h_{z,r})(0), \dots, (\Delta^{d-1}h_{z,r})(0)$ is non-zero mod $p^{\nu}$, so to determine $j(h_{z,r})$ it suffices to compute these quantities mod $p^\nu$. We have the formula
\begin{align*}
(\Delta^k h_{z,r})(0) = \sum_{n=0}^k (-1)^{k-n} \binom{k}{n} h_{z,r}(n) = \sum_{n=0}^k (-1)^{k-n} \binom{k}{n} u_{M(p^rn + z) + \ell} \, .
\end{align*}
Thus $(\Delta^k h_{z,r})(0) \mod p^\nu$ may be computed in $\poly(\|\LRS{u}\|)$ time if $u_{M(p^rn + z) + \ell} \mod p^\nu$ may be computed in $\poly(\|\LRS{u}\|)$ time, for $0 \leq n \leq d-1$. Writing $\alpha = \begin{pmatrix} 0 & 0 \dots & 1 \end{pmatrix}$, $\beta = \begin{pmatrix} u_{d-1} & u_{d-2} \dots & u_0 \end{pmatrix}^T$ and
\begin{align*}
A = \begin{pmatrix} a_{d-1} & \dots & a_1 & a_0 \\ 1 & \dots & 0 & 0 \\ \vdots & \ddots & \vdots & 0 \\ 0 & \dots  & 1 & 0  \end{pmatrix}
\end{align*}
the companion matrix, we have
\begin{align*}
    u_m = \alpha A^m \beta \, .
\end{align*}
For $0 \leq n \leq d-1$, we have $M(p^r n + z) + \ell \leq q_1(\|\LRS{u}\|)^{q_2(r)}$ for polynomials $q_1,q_2$. We may compute $A^{M(p^r n + z)} \mod p^\nu$ using $O(\log(M(p^rn+z) + \ell)) = O(q_2(r)\log(q_1(\|\LRS{u}\|))) = \poly(\|\LRS{u}\|,r)$ multiplications of matrices mod $p^\nu$, using repeated squaring. Since $\nu = \poly(\|\LRS{u}\|,r)$ and $p = \poly(\|\LRS{u}\|)$, the entries of the matrices mod $p^\nu$ may be written with $\poly(\|\LRS{u}\|,r)$ many binary digits, so each multiplication of matrices may be done in $\poly(\|\LRS{u}\|,r)$ time. 

In conclusion, $(\Delta^k h_{z,r})(0) \mod p^\nu$ may be computed in $\poly(\|\LRS{u}\|)$ time for all $0 \leq k \leq d-1$, from which we can immediately find $j(h_{z,r})$ and conclude. 
\end{proof}
\begin{corollary} \label{cor:fast_discs}
For each $0 \leq \ell \leq M-1$, for integer $N$ written in binary, we can take $R = O(\|N\|)$ such that $Mp^R + \ell \geq N$ and we may find all discs $\overline D(z,R)$ containing a zero $x_0 \in \O_p$ of $F_\ell$ in $\poly(\|\LRS{u}\|,\|N\|)$ time.
\end{corollary}
\begin{proof}
Firstly, note that we can take $R = \left\lceil \frac{\log N}{\log p} \right\rceil$ so $R = O(\|N\|)$. By Proposition \ref{prop:fast_disc}, for each integer $1 \leq r \leq R$ and each integer $0 \leq z \leq p^r-1$ it takes $\poly(\|\LRS{u}\|,r) = \poly(\|\LRS{u}\|,\|N\|)$ time to determine whether $\overline{D}(z,r)$ contains any zeros of $F_\ell$. Consider the following routine:
\begin{enumerate}
    \item Determine whether $F_\ell$ has any zeros in $\overline{D}(0,0) = \O_p$.
    \item Given a disc $\overline{D}(z,r)$ that has been determined to have zeros of $F_\ell$, for each integer $a = 0,1,\dots ,p-1$, determine whether $\overline{D}(z+p^ra,r+1)$ has any zeros of $F_\ell$. 
    \item Repeat step 2 until $r = R$, output all discs $\overline D(z,R)$ found and then stop.
\end{enumerate}
This routine inductively finds all the discs $\overline D(z,R)$ containing zeros of $F_\ell$ for $0 \leq z \leq p^R-1$. By Theorem \ref{thm:interpolate}, $F_\ell$ has at most $d-1$ zeros, so we end up with at most $d-1$ of these discs. Each such disc $\overline D(z,R)$ takes at most $pR = \poly(\|\LRS{u}\|,\|N\|)$ operations of checking particular discs $\overline D(z',r)$ for zeros, and each operation takes $\poly(\|\LRS{u}\|,\|N\|)$ time, so the entire routine takes $\poly(\|\LRS{u}\|,\|N\|)$ time.
\end{proof}
We are now ready to prove our main result.
\SkolemEqSLP*
\begin{proof}
We can find $p,M = \poly(\|\LRS{u}\|)$ in $\poly(\|\LRS{u}\|)$-time satisfying the conditions of Theorem \ref{thm:interpolate} by Lemma \ref{lem:prime_and_M}. For each subsequence $u_{Mn+\ell}$ for $0 \leq \ell \leq M-1$ we check whether $u_{Mn+\ell}$ is identically zero by checking directly whether $u_{Mn+\ell} = 0$ for $n = 0,1,\dots,d-1$. This takes $\poly(\|\LRS{u}\|)$ time as $Mn+\ell = \poly(\|\LRS{u}\|)$ for such values of $n$. For each $\ell$ such that $u_{Mn+\ell}$ is not identically zero we analyse its $p$-adic interpolation $F_\ell : \Z_p \to \Z_p$.

By Corollary \ref{cor:fast_discs}, for each $0 \leq \ell \leq M-1$ and for $R = O(\|N\|)$ such that $Mp^R + \ell \geq N$ we may find all discs $\overline D(z,R)$ for integers $0 \leq z \leq p^R-1$ containing a zero of $F_\ell$. By construction of $R$, the only positive integer inside each disc $\overline{D}(z,R)$ possibly corresponding to an integer zero $0 \leq n \leq N$ of $\LRS{u}$ is $z$ itself. Therefore, for each such disc we need only check whether $u_{Mz+\ell} = 0$.

Let the collection of all candidate integer zeros be $n_1,\dots, n_m$. Note each $n_i = Mz_i + \ell$ for some $0 \leq \ell \leq M-1$ and $0 \leq z_i \leq p^R$. By Theorem \ref{thm:interpolate}, each $F_\ell$ has at most $d-1$ zeros, so $m \leq (d-1)M = \poly(\|\LRS{u}\|)$. However, we cannot check whether $u_{n_i} = 0$ directly since in general $n_i = 2^{\poly(\log\|\LRS{u}\|,\|N\|)}$, so one more trick is required. Consider the product $U = u_{n_1} u_{n_2} \dots u_{n_m}$. Since we have
\begin{align*}
u_{n_i} = \alpha A^{n_i} \beta
\end{align*}
for certain vectors $\alpha,\beta$ and the companion matrix $A$, each $u_{n_i}$ may be written as the output of a $\poly(\|\LRS{u}\|,\|N\|)$-sized circuit, using repeated squaring. Therefore, since $m = \poly(\|\LRS{u}\|)$, we may write $U$ as the output of a $\poly(\|\LRS{u}\|,\|N\|)$-sized circuit. The LRS $\LRS{u}$ has a zero $0 \leq n \leq N$ if and only if $U = 0$, and the problem of checking whether $U = 0$ is in $\mathsf{EqSLP}$.
\end{proof}
Since the described algorithm identifies all possible candidates for zeros $0 \leq n \leq N$ of $\LRS{u}$, we can improve the complexity of the function problem version of the Bounded Skolem Problem. 
\begin{theorem} \label{thm:function_prob}
Let $d\in \mathbb N$.
Given $\LRS{u} \in \mathrm{LRS}(d)$ and $N\in\mathbb N$, the problem of computing the set of all $0\leq n \leq N$ with $u_n=0$ (represented as the union of a finite set and finitely many arithmetic progressions) is in 
$\mathsf{FP}^{\mathsf{EqSLP}} = \mathsf{FP}^{\mathsf{RP}}$.
%arithmetic progressions $(Mn + \ell)_{n \in \N}$ for which $u_{Mn+\ell} = 0$ for all $n \in \N$ and all integers $0 \leq n \leq N$ such that $u_n = 0$ is in $\mathsf{FP}^{\mathsf{EqSLP}} = \mathsf{FP}^{\mathsf{RP}}$.
\end{theorem}
\begin{proof}
We may follow the same exact proof as for Theorem \ref{thm:Skolem_EqSLP}, except using an $\mathsf{EqSLP}$ oracle to decide whether $u_n = 0$ for each candidate zero $n$.
\end{proof}
\begin{corollary}
The Skolem Problem for LRS of order at most 4 is disjunctively Turing reducible to $\mathsf{EqSLP}$ and hence lies in $\mathsf{coRP}$. The function problem of determining all zeros of an LRS of order at most 4 is in $\mathsf{FP}^{\mathsf{EqSLP}} = \mathsf{FP}^{\mathsf{RP}}$. 
\end{corollary}
\begin{proof}
Follows from Theorems \ref{thm:Skolem_EqSLP} and \ref{thm:function_prob} and the fact that there is an upper bound $N = 2^{\poly(\|\LRS{u}\|)}$ on the size of the largest zero lying outside an arithmetic progression of zeros. We deduce this from \cite[Appendix C]{Orbit_problem_2016} which states that a \emph{non-degenerate} sequence $\LRS{u}$ of order at most 4 has an upper bound $N = 2^{O(\|\LRS{u}\|')}$ on the size of the largest zero, where $\|\LRS{u}\|'$ is the size of a description of the exponential-polynomial representation of $\LRS{u}$ in binary. That is, if
\begin{align*}
u_n = \sum_{i=1}^s P_i(n)\lambda_i^n
\end{align*}
for polynomials $P_i$ with algebraic coefficients and algebraic numbers $\lambda_i$, then $\|\LRS{u}\|' = \sum_i \|P_i\| + \|\lambda_i\|$. The description of the coefficients of $P_i$ and $\lambda_i$ are the standard ways to represent an algebraic number by its minimal polynomial and an appropriate approximation, see \cite{Orbit_problem_2016} for more details. It is noted there also that there is an integer $L(d)$ depending only on the order $d$ of the LRS such that the subsequences $u_{L(d)n+r}$ are identically zero or non-degenerate. These subsequences have exponential polynomial representation
\begin{align*}
u_{L(d)n+r} = \sum_{i=1}^s P_i(L(d)n+r) \lambda_i^r \lambda_i^{L(d)n} 
\end{align*}
and by \cite[Lemma A.1]{Orbit_problem_2016} we have $\|Q_i\|,\|\lambda_i^{L(d)}\| = \poly(\|\LRS{u}\|')$ where $Q_i(n) = P_i(L(d)n+r)$, so each subsequence $v_n^{(r)} = u_{L(d)n+r}$ has $\|v^{(r)}\|' = \poly(\|\LRS{u}\|')$. Thus, since $\|\lambda_i\|,\|P_i\| = \poly(\|\LRS{u}\|)$ (e.g. see \cite[Lemma 6]{akshay2020}) we have $\|\LRS{u}\|' = \poly(\|\LRS{u}\|)$. By \cite{Orbit_problem_2016} there's a $N_r = 2^{O(\|\LRS{v}^{(r)}\|'}$ for each not-identically-zero $\LRS{v}^{(r)}$ so we may take $N = \underset{0\leq r \leq L(d)-1} \max N_r = 2^{\poly(\|\LRS{u}\|)}$.
\end{proof}
In fact, the Skolem Problem is known to be decidable for a wider class of LRS\@. Named after the authors of the papers \cite{Mignotte_distance,Vereshchagin} in which decidability is established, the \emph{MSTV class} is the class of LRS with at most 3 dominant roots with respect to an Archimedean absolute value or at most 2 dominant roots with respect to a non-Archimedean absolute value. See also the exposition of \cite{bilu2025} on these results. Let $\mathrm{MSTV}(d)$ denote the class of LRS of order $d$ lying in the MSTV class. We claim that by carrying the constants through more precisely in the proof of decidability, there is a bound $N = 2^{\poly(\|\LRS{u}\|)}$ on the size of the largest zero of any non-degenerate LRS in $\mathrm{MSTV}(d)$, and hence for any integer $d \geq 1$ the Skolem Problem for the class $\mathrm{MSTV}(d)$ is disjunctively Turing reducible to $\mathsf{EqSLP}$ and therefore in $\mathsf{coRP}$.

Finally, one more class of LRS for which an exponential upper bound on the size of the largest zero was shown is the class of LRS whose characteristic roots are all 
roots of rational numbers~\cite{akshay2020}.  For this class we likewise obtain a \textsf{coRP} bound to decide the Skolem Problem for LRS of every fixed order.

\subsection{Pseudocode}
For completeness, we give the full pseudocode for the algorithm to find candidate zeros of $\LRS{u}$ previously described. The algorithm takes as input an LRS $\LRS{u}$ and integer $N$ and finds in polynomial-time a description of a set containing all possible integer zeros of $\LRS{u}$ in the interval $[0,N]$. In particular this set is comprised of finitely many arithmetic progressions $(Mn+\ell)_{n \in \N}$ for which $u_{Mn+\ell}=0$, and polynomially-many exceptional integers $0 \leq n \leq N$ lying outside those arithmetic progressions. To solve the Bounded Skolem Problem, one simply has to check whether the product $u_{n_1}\dots u_{n_s} = 0$ where $n_1, \dots, n_s$ are all the exceptional candidate zeros. To solve the function problem version, one checks whether $u_{n_j} = 0 $ for each $j=1,\dots,s$.

\algrenewcommand\algorithmicrequire{\textbf{Input:}}
\algrenewcommand\algorithmicensure{\textbf{Output:}}

\begin{algorithm}
\caption{Polynomial-time algorithm for finding all arithmetic progressions of zeros and polynomially-many extra candidate integer zeros of an LRS}
\label{alg:candidates}
\begin{algorithmic}[1]

  \Require LRS $\LRS{u}$ satisfying \eqref{eqn:LRS_def} of order $d$, integer $N$
  \Ensure List of pairs $(M,\ell)$ for which $u_{Mn+\ell}$ is identically zero, and a list of polynomially many exceptional candidate zeros $0 \le n \le N$

  \Procedure{ZeroSearch}{$F,d,N,\nu,M,\ell,p$}
    \State $\texttt{discs} \gets$ empty list
    \State $\texttt{zeros} \gets$ empty list
    \State Compute $(\Delta^k F)(0) \bmod p^\nu$ for each $0 \le k \le d-1$
    \State From these, compute $\min\{v_p((\Delta^k F)(0)),\nu\}$ for each $0 \le k \le d-1$
    \State $j(F) \gets \text{largest } k \text{ with } v_p((\Delta^k F)(0)) = \underset{0 \le j \le d-1}{\min} v_p((\Delta^j F)(0))$
    \If{$j(F) > 0$}
      \State append $\overline D(0,0)$ to $\texttt{discs}$
    \EndIf
    \While{$\texttt{discs} \neq \emptyset$}
      \ForAll{$\overline D(z,r) \in \texttt{discs}$}
        \ForAll{$a \in \{0,1,\dots,p-1\}$}
          \State $z' \gets z + p^r a$
          \State Compute $(\Delta^k h_{z',r+1})(0) \bmod p^\nu$ for each $0 \le k \le d-1$
          \State From these, compute $\min\{v_p((\Delta^k h_{z',r+1})(0)),\nu\}$ for each $0 \le k \le d-1$
          \State $j(h_{z',r+1}) \gets \text{largest } k \text{ with } v_p((\Delta^k h_{z',r+1})(0)) = \underset{0 \le j \le d-1}{\min} v_p((\Delta^j h_{z',r+1})(0))$
          \If{$j(h_{z',r+1}) \geq 1$}
            \If{$p^{r+1}M + \ell \ge N$}
              \State append $z'$ to $\texttt{zeros}$
            \Else
              \State append $\overline D(z',r+1)$ to $\texttt{discs}$
            \EndIf
          \EndIf
        \EndFor
        \State remove $\overline D(z,r)$ from $\texttt{discs}$
      \EndFor
    \EndWhile
    \State \Return $\texttt{zeros}$
  \EndProcedure

  \State $\texttt{candidates} \gets$ empty list
  \State $\texttt{APs} \gets$ empty list
  \State Compute the smallest prime $p$ with $p \ge d+2$ and $p \nmid a_0$
  \State $A \gets$ companion matrix of $\LRS{u}$
  \State $M \gets$ the smallest positive integer such that $A^M \equiv I \pmod p$
  \State $\nu \gets \log\!\big(2^d\, 2^{(Md+k+1)\|\LRS{u}\|}\, d^{Mn+k+1}\big) + rd + 1$
  \ForAll{$\ell \in \{0,1,\dots,M-1\}$}
    \If{$u_{Mn+\ell} = 0$ for every $n = 0,1,\dots,d-1$}
      \State append $(M,\ell)$ to $\texttt{APs}$
    \Else
      \State $F_\ell \gets \text{the function } \big(n \mapsto u_{Mn+\ell}\big)$ on $\mathbb{N}$
      \ForAll{$z \in \Call{ZeroSearch}{F_\ell,d,N,\nu,M,\ell,p}$}
        \State append $Mz+\ell$ to $\texttt{candidates}$
      \EndFor
    \EndIf
  \EndFor
  \State \textbf{output } $\texttt{candidates}$, $\texttt{APs}$

\end{algorithmic}
\end{algorithm}

\newpage
\subsection{An example}
Here we show an example of how the algorithm works for a particular LRS\@. Note this is only supposed to be illustrative so $\nu$ is chosen smaller than in Algorithm \ref{alg:candidates}. Let $\LRS{u}$ be defined by
\begin{align*}
    u_{n+3} = 2u_{n+2} -3u_{n+1} + u_n
\end{align*}
and initial values $u_0 = -1, u_1 = 1, u_2 = 7$. We may take prime $p = 5$, and we can see that the companion matrix $A$ satisfies $A^8 \equiv I \mod p$, thus we can take $M=8$. Suppose that we are also given upper bound $N = 200$. For the benefit of presentation, we take $\nu = 5$.
Following the algorithm, consider the subsequences $u_{8n},u_{8n+1},\dots,u_{8n+7}$ and their corresponding $p$-adic analytic interpolations $F_0,F_1,\dots,F_7$. It is easy to check that none of the subsequences are identically zero.
To compute the number of zeros of $F_\ell$ in $\overline D(0,0)$, we compute each of $(\Delta^0F_\ell)(0),(\Delta^1 F_\ell)(0), (\Delta^2F_\ell)(0) \mod 5^5$. 

\begin{table}[ht] \label{table:F_ell}
\centering
\begin{tabular}{|r|r|r|r|}
\hline
$\ell$ & $(\Delta^0 F_\ell)(0) $ & $(\Delta^1 F_\ell)(0)$ & $(\Delta^2 F_\ell)(0)$ \\ \hline
0 & 3124 & 80 & 400 \\ \hline
1 & 1 & 130 & 2875 \\ \hline 
2 & 7 & 15 & 25 \\ \hline 
3 & 10 & 2845 & 1200 \\ \hline
4 & 0 & 2650 & 2075 \\ \hline
5 & 3102 & 3030 & 575 \\ \hline
6 & 3089 & 955 & 2375 \\ \hline 
7 & 3122 & 1720 & 1975 \\ \hline
\end{tabular}
\caption{The values of $(\Delta^k F)(0) \mod 5^5$ for $0 \leq \ell \leq 7$}
\end{table}

\begin{table}[ht]
\centering
\begin{tabular}{|r?r|r?r|r|}
\hline
$z$ & $(\Delta^0 h_{z,1})(0) $ & $(\Delta^1 h_{z,1})(0)$ & $(\Delta^0 h_{2_5z,2})(0)$ & $(\Delta^1 h_{2+5z,2})(0)$ \\ \hline
0 & 10 & 2475 & 650 & 500  \\ \hline
1 & 2855 & 350 & 2625 & 500 \\ \hline 
2 & 650 & 1975 & 225 & 500 \\ \hline 
3 & 395 & 1100 & 2825 & 500 \\ \hline
4 & 2215 & 850 & 1050 & 500 \\ \hline
\end{tabular}
\caption{The values of $(\Delta^k h_{z,1})(0) \mod 5^5$ for $0 \leq z \leq 4$ where $h_{z,1}(x) = F_3(px+z)$}
\label{table:F3}
\end{table}
\begin{table}[ht]
\centering
\resizebox{\columnwidth}{!}{%
\begin{tabular}{|r?r|r|r?r|r?r|r|}
\hline
$z$ & $(\Delta^0 h_{z,1})(0) $ & $(\Delta^1 h_{z,1})(0)$ & $(\Delta^2 h_{z,1})(0)$ & $(\Delta^0 h_{5z,2})(0)$ & $(\Delta^1 h_{5z,2})(0)$ & $(\Delta^0 h_{2+5z,2})(0)$ & $(\Delta^1 h_{2+5z,2})(0)$ \\ \hline
0 & 0 & 2750 & 1875 & 0 & 1250 & 1125 & 1875 \\ \hline
1 & 2650 & 625 & 1875 & 2750 & 1250 & 2750 & 1875 \\ \hline 
2 & 1125 & 1625 & 1875 & 1125 & 1250 & 0 & 1875 \\ \hline 
3 & 2300 & 2625 & 1875 & 1375 & 1250 & 2250 & 1875 \\ \hline
4 & 1800 & 500 & 1875 & 375 & 1250 & 125 & 1875 \\ \hline
\end{tabular}
}
\caption{The values of $(\Delta^k h_{z,1})(0) \mod 5^5$ for $0 \leq z \leq 4$ where $h_{z,1}(x) = F_4(px+z)$}
\label{table:F4}
\end{table}

For $\ell = 0,1,2,5,6,7$ we have $v_5(\Delta^0 F_\ell)(0) < v_5(\Delta^1 F_\ell )(0),v_5(\Delta^2 F_\ell)(0)$ so $j(F_\ell) = 0$, therefore $F_\ell$ has no zeros in $\overline D(0,0)$ for these values of $\ell$. For $\ell = 3,4$ we see that $j(F_3) = 1$ and $j(F_4) = 2$ so $F_3,F_4$ have 1 and 2 zeros in $\overline D(0,0)$ respectively. 

Next, we find the zeros of $F_3$. Since $j(F_3) =1$, only $(\Delta^0 h_{z,r})(0),(\Delta^1 h_{z,r})(0)$ need to be computed to determine where the zero of $F_3$ lies, so for presentation $(\Delta^2 h_{z,r})(0)$ has been suppressed from the tables. Table \ref{table:F3} shows the only $z$ for which $v_p((\Delta^0 h_{z,1})(0)) \geq v_p((\Delta^1 h_{z,1})(0))$ is $z = 2$ so the zero of $F_3$ lies in $\overline{D}(1,2)$. Similarly, $v_p((\Delta^0 h_{7,2})(0)) \geq v_p((\Delta^1 h_{7,2})(0))$ so the zero of $F_3$ lies in $\overline D(7,2)$. Since $N=200$ and $Mp^2 + \ell = 203 \geq 200$, there is only one possible candidate integer zero of $\LRS{u}$ corresponding to this; we need only check $u_{8\cdot 7 + 3} = u_{59}$.

We repeat the same process with $F_4$. Analysing the values of $(\Delta^k h_{z,r})(0)$ given in Table \ref{table:F4} shows $\overline D(0,1) \supset \overline D(0,2)$ and $\, \overline D(2,1) \supset \overline D(12,2)$ contain one zero of $F_4$ each. Again, since $Mp^2 + \ell = 204 \geq 200$, we need only check $u_{8\cdot 0 + 4} = u_4$ and $u_{8\cdot 12 + 4} = u_{100}$.

At this point, to decide the Skolem Problem we use an algorithm for $\mathsf{EqSLP}$ to decide whether $u_{4}u_{59}u_{100} = 0$. To solve the function problem of determining all zeros, call an $\mathsf{EqSLP}$ oracle to check whether each $u_4,\, u_{59}, \, u_{100} = 0$. It turns out in this case that the only zero is $u_4 = 0$.
\section{Conclusion}
We have exhibited an algorithm that shows that for every $d\in \mathbb N$
the Bounded Skolem Problem lies in $\mathsf{coRP}$ for LRS of order at most $d$. 
As a corollary we showed that the Skolem Problem for LRS of order 4 lies also lies in $\mathsf{coRP}$.
The most natural route to prove decidability of Skolem's Problem is to exhibit an effective upper bound 
for the largest integer zero of a non-degenerate LRS (as is done in the order-4 case).  It is plausible that such a bound has magnitude at most exponential in the description of the LRS---there is no known family of LRS for which the largest zero is super-exponential.  The existence of such an exponential bound would yield a polynomial-time reduction of Skolem's Problem to the Bounded Skolem Problem.

\section*{Acknowledgements.}
Jo\"{e}l Ouaknine is also affiliated with Keble College, Oxford as \href{https://emmy.network}{emmy.network} Fellow, and supported by ERC grant DynAMiCs (101167561) and DFG grant 389792660 as part of \href{https://perspicuous-computing.science}{TRR~248}. Piotr Bacik and James Worrell were supported by EPSRC grant EP/X033813/1.

\bibliographystyle{siamplain}
\bibliography{main.bib}

@book{Tao08,
author = {T. Tao},
title = {Structure and Randomness},
publisher = {American Mathematical Society},
year = {2008}
}

@book{SalomaaS78,
   author = {Salomaa, A. and Soittola, M.},
  publisher = {Springer},
  series = {Texts and monographs in computer science},
  title = {Automata-theoretic aspects of formal power series.},
  year = 1978
}

@article{rosser_1962,
	title = {Approximate formulas for some functions of prime numbers},
	volume = {6},
	issn = {0019-2082},
	doi = {10.1215/ijm/1255631807},
	number = {1},
	journal = {Illinois Journal of Mathematics},
	author = {Rosser, J. Barkley and Schoenfeld, Lowell},
	month = mar,
	year = {1962},
}

@article{van_der_poorten_1991, 
title={Zeros of recurrence sequences}, 
volume={44}, 
DOI={10.1017/S0004972700029646},
number={2}, 
journal={Bulletin of the Australian Mathematical Society}, 
author={van der Poorten, A.J. and Schlickewei, H.P.}, 
year={1991}, 
pages={215–223}
}

@misc{bilu2025,
      title={Skolem Problem for Linear Recurrence Sequences with 4 Dominant Roots (after Mignotte, Shorey, Tijdeman, Vereshchagin and Bacik)}, 
      author={Y. Bilu},
      year={2025},
      eprint={2501.16290},
      archivePrefix={arXiv},
      primaryClass={math.NT},
}

@article{Vereshchagin,
	title = {Occurrence of zero in a linear recursive sequence},
	volume = {38},
	copyright = {http://www.springer.com/tdm},
	issn = {0001-4346, 1573-8876},
	doi = {10.1007/BF01156238},
	language = {en},
	number = {2},
	journal = {Mathematical Notes of the Academy of Sciences of the USSR},
	author = {Vereshchagin, N. K.},
	year = {1985},
	pages = {609--615},
}

@article{Mignotte_distance,
	title = {The distance between terms of an algebraic recurrence sequence.},
	volume = {1984},
	issn = {0075-4102, 1435-5345},
	doi = {10.1515/crll.1984.349.63},
	number = {349},
	journal = {Journal für die reine und angewandte Mathematik (Crelles Journal)},
	author = {Tijdeman, R. and Mignotte, M. and Shorey, T.N.},
	month = may,
	year = {1984},
	pages = {63--76},
}

@article{Orbit_problem_2016,
author = {Chonev, V. and Ouaknine, J. and Worrell, J.},
title = {On the Complexity of the Orbit Problem},
year = {2016},
issue_date = {September 2016},
publisher = {Association for Computing Machinery},
address = {New York, NY, USA},
volume = {63},
number = {3},
issn = {0004-5411},
doi = {10.1145/2857050},
journal = {J. ACM},
month = jun,
articleno = {23},
numpages = {18}
}

@InProceedings{akshay2020,
  author =	{Akshay, S. and Balaji, N. and Murhekar, A. and Varma, R. and Vyas, N.},
  title =	{{Near-Optimal Complexity Bounds for Fragments of the Skolem Problem}},
  booktitle =	{37th International Symposium on Theoretical Aspects of Computer Science (STACS 2020)},
  pages =	{37:1--37:18},
  series =	{Leibniz International Proceedings in Informatics (LIPIcs)},
  ISBN =	{978-3-95977-140-5},
  ISSN =	{1868-8969},
  year =	{2020},
  volume =	{154},
  publisher =	{Schloss Dagstuhl -- Leibniz-Zentrum f{\"u}r Informatik},
  address =	{Dagstuhl, Germany},
  doi =		{10.4230/LIPIcs.STACS.2020.37},
}

@article{Blondel2002,
title = {The presence of a zero in an integer linear recurrent sequence is NP-hard to decide},
journal = {Linear Algebra and its Applications},
volume = {351-352},
pages = {91-98},
year = {2002},
note = {Fourth Special Issue on Linear Systems and Control},
issn = {0024-3795},
doi = {https://doi.org/10.1016/S0024-3795(01)00466-9},
author = {Vincent D. Blondel and Natacha Portier},
}

@InProceedings{Akshay2017,
  author =	{S., Akshay and Balaji, Nikhil and Vyas, Nikhil},
  title =	{{Complexity of Restricted Variants of Skolem and Related Problems}},
  booktitle =	{42nd International Symposium on Mathematical Foundations of Computer Science (MFCS 2017)},
  pages =	{78:1--78:14},
  series =	{Leibniz International Proceedings in Informatics (LIPIcs)},
  ISBN =	{978-3-95977-046-0},
  ISSN =	{1868-8969},
  year =	{2017},
  volume =	{83},
  publisher =	{Schloss Dagstuhl -- Leibniz-Zentrum f{\"u}r Informatik},
  address =	{Dagstuhl, Germany},
  URN =		{urn:nbn:de:0030-drops-81306},
  doi =		{10.4230/LIPIcs.MFCS.2017.78},
}

@article{Ouaknine_loop2015,
author = {Ouaknine, Jo\"{e}l and Worrell, James},
title = {On linear recurrence sequences and loop termination},
year = {2015},
issue_date = {April 2015},
publisher = {Association for Computing Machinery},
address = {New York, NY, USA},
volume = {2},
number = {2},
doi = {10.1145/2766189.2766191},
journal = {ACM SIGLOG News},
month = apr,
pages = {4–13},
numpages = {10}
}

@article{almagor_deciding_2021,
	title = {Deciding $\omega$-regular properties on linear recurrence sequences},
	volume = {5},
	doi = {10.1145/3434329},
	number = {POPL},
	journal = {Proc. ACM Program. Lang.},
	author = {Almagor, Shaull and Karimov, Toghrul and Kelmendi, Edon and Ouaknine, Joël and Worrell, James},
	month = jan,
	year = {2021},
	pages = {48:1--48:24},
}

@book{berstel_noncommutative_2010,
	edition = {1},
	title = {Noncommutative {Rational} {Series} with {Applications}},
	copyright = {https://www.cambridge.org/core/terms},
	isbn = {9780511760860},
	publisher = {Cambridge University Press},
	author = {Berstel, Jean and Reutenauer, Christophe},
	month = {10},
	year = {2010},
    pages = {130},
	doi = {10.1017/CBO9780511760860},
}

@article{bell_mortality_2021,
	title = {On the mortality problem: {From} multiplicative matrix equations to linear recurrence sequences and beyond},
	volume = {281},
	issn = {08905401},
	shorttitle = {On the mortality problem},
	doi = {10.1016/j.ic.2021.104736},
	language = {en},
	journal = {Information and Computation},
	author = {Bell, Paul C. and Potapov, Igor and Semukhin, Pavel},
	month = dec,
	year = {2021},
	pages = {104736},
}

@article{agrawal_approximate_2015,
	title = {Approximate {Verification} of the {Symbolic} {Dynamics} of {Markov} {Chains}},
	volume = {62},
	issn = {0004-5411},
	doi = {10.1145/2629417},
	number = {1},
	journal = {J. ACM},
	author = {Agrawal, Manindra and Akshay, S. and Genest, Blaise and Thiagarajan, P. S.},
	month = mar,
	year = {2015},
	pages = {2:1--2:34},
}

@inproceedings{barthe_universal_2020,
	address = {New York, NY, USA},
	series = {{LICS} '20},
	title = {Universal equivalence and majority of probabilistic programs over finite fields},
	isbn = {9781450371049},
	doi = {10.1145/3373718.3394746},
	booktitle = {Proceedings of the 35th {Annual} {ACM}/{IEEE} {Symposium} on {Logic} in {Computer} {Science}},
	publisher = {Association for Computing Machinery},
	author = {Barthe, Gilles and Jacomme, Charlie and Kremer, Steve},
	month = jul,
	year = {2020},
	pages = {155--166},
}

@article{blondel_survey_2000,
	title = {A survey of computational complexity results in systems and control},
	volume = {36},
	copyright = {https://www.elsevier.com/tdm/userlicense/1.0/},
	issn = {00051098},
	doi = {10.1016/S0005-1098(00)00050-9},
	language = {en},
	number = {9},
	journal = {Automatica},
	author = {Blondel, V. D. and Tsitsiklis, J. N.},
	month = sep,
	year = {2000},
	pages = {1249--1274},
}

@Misc{Skolem_SML,
 Author = {Skolem, Th.},
 Title = {Ein {Verfahren} zur {Behandlung} gewisser exponentialer {Gleichungen} und diophantischer {Gleichungen}},
 Year = {1935},
 Language = {German},
 HowPublished = {8. {Skand}. {Mat}.-{Kongr}., 163-188 (1935).},
 zbMATH = {3017771},
 Zbl = {0011.39201}
}

@article{lech_note_1953,
	title = {A note on recurring series},
	volume = {2},
	issn = {0004-2080},
	doi = {10.1007/BF02590997},
	language = {en},
	number = {5},
	journal = {Arkiv för Matematik},
	author = {Lech, Christer},
	month = {8},
	year = {1953},
	pages = {417--421},
}

@Article{Mahler_SML,
 Author = {Mahler, K.},
 Title = {Eine arithmetische {Eigenschaft} der {Taylor}-{Koeffizienten} rationaler {Funktionen}.},
 FJournal = {Proceedings. Akadamie van Wetenschappen Amsterdam},
 Journal = {Proc. Akad. Wet. Amsterdam},
 ISSN = {0370-0348},
 Volume = {38},
 Pages = {50--60},
 Year = {1935},
 Language = {German},
 zbMATH = {2531698},
 JFM = {61.0176.02}
}

@book{recurrence_sequences,
	address = {Providence, RI},
	series = {Mathematical surveys and monographs},
	title = {Recurrence sequences},
	isbn = {9780821833872},
	number = {v. 104},
	publisher = {American Mathematical Society},
	author = {Everest, Graham and Van der Poorten, Alf and Shparlinski, Igor and Ward, Thomas},
	year = {2003},
    pages = {5},
}

@InProceedings{Schonhage_1979,
author="Sch{\"o}nhage, Arnold",
editor="Maurer, Hermann A.",
title="On the power of random access machines",
booktitle="Automata, Languages and Programming",
year="1979",
publisher="Springer Berlin Heidelberg",
address="Berlin, Heidelberg",
pages="520--529",
isbn="978-3-540-35168-9"
}

@book{robert_course_2000,
	address = {New York, NY},
	series = {Graduate texts in mathematics},
	title = {A course in p-adic analysis},
	isbn = {9781441931504},
	language = {eng},
	number = {198},
	publisher = {Springer},
	author = {Robert, A.},
	year = {2000},
doi={10.1007/978-1-4757-3254-2}
}

@book{gouvea_p-adic_2020,
	address = {Cham},
	series = {Universitext},
	title = {p-adic {Numbers}: {An} {Introduction}},
	copyright = {http://www.springer.com/tdm},
	isbn = {9783030472948},
	shorttitle = {p-adic {Numbers}},
	language = {en},
	publisher = {Springer International Publishing},
	author = {Gouvêa, Fernando Q.},
	year = {2020},
	doi = {10.1007/978-3-030-47295-5},
}
\end{document}